\documentclass[12pt,reqno]{amsart}

\usepackage{amsaddr,amssymb}
\usepackage{graphicx} 
\usepackage{caption}
\usepackage{subcaption}

\newtheorem{claim}{Claim}[section]
\newtheorem{theorem}[claim]{Theorem}

\theoremstyle{definition}
\newtheorem{remark}[claim]{Remark}

\newtheorem{definition}[claim]{Definition}

\newcommand{\soutg}{\bgroup\markoverwith{\textcolor{green}{\rule[.5ex]{2pt}{1pt}}}\ULon}
\newcommand{\soutb}{\bgroup\markoverwith{\textcolor{blue}{\rule[.5ex]{2pt}{1pt}}}\ULon}
\newcommand{\soutr}{\bgroup\markoverwith{\textcolor{red}{\rule[.5ex]{2pt}{1pt}}}\ULon}

\numberwithin{equation}{section}

\begin{document}
\title[Determinants for non-selfadjoint operators]{The role of the branch cut of the logarithm in the definition of the spectral determinant for non-selfadjoint operators}

\author{Ji\v{r}\'{\i} Lipovsk\'{y}}
\address{Department of Physics, Faculty of Science, University of Hradec Kr\'alov\'e, Rokitansk\'eho 62,
500\,03 Hradec Kr\'alov\'e, Czechia}
\email{jiri.lipovsky@uhk.cz}

\author{Tom\'a\v{s} Mach\'a\v{c}ek} 
\address{Department of Physics, Faculty of Science, University of Hradec Kr\'alov\'e, Rokitansk\'eho 62,
500\,03 Hradec Kr\'alov\'e, Czechia}
\email{Machy.17@seznam.cz}

\date{\today}

\begin{abstract}
The spectral determinant is usually defined using the spectral zeta function that is meromorphically continued to zero. In this definition, the complex logarithms of the eigenvalues appear. Hence the notion of the spectral determinant depends on the way how one chooses the branch cut in the definition of the logarithm. We give results for the non-self-adjoint operators that state when the determinant can and cannot be defined and how its value differs depending on the choice of the branch cut.
\end{abstract}
\maketitle

\smallskip
\noindent \emph{Keywords:} spectral determinant; branch cut; spectral zeta function.

\section{Introduction}
When studying the spectral properties of operators, various notions can be investigated. One of them is the spectral (functional) determinant, corresponding to the product of eigenvalues. It can be viewed as a generalization of the determinant for a square matrix (an operator with finitely many eigenvalues). Since for most of the interesting operators the product of their eigenvalues is not convergent, one defines it using the spectral zeta function, a function of a complex variable $s$ defined through an infinite sum that is usually convergent in a certain half-plane in $s$ (for details see Section~\ref{sec:def}). Using the fact that the spectral zeta function can be uniquely meromorphically continued to the rest of the complex plane in $s$, allows us to assign the unique value to its derivative at $s=0$. This value is then used for the definition of the determinant.

The above definition of the spectral determinant can be traced back to the works of Minakshisundaram and Pleijel \cite{MP49} and Ray and Singer \cite{RS71}. Since then, results for various operators have been obtained. Without claiming that the list of works is complete, we mention e.g. the papers on the determinant for the Sturm-Liouville operators \cite{LS77, GK19, ACF20}, Dirichlet Laplacians on balls or polygons \cite{BGKE96, AS94}, or harmonic and anharmonic oscillators \cite{Fre18, Vor80}. An important application of the spectral determinants can be found in string theory or quantum field theory (see, \cite{Dun08} and references therein). A result for more general elliptic operators obtained in \cite{BFK95} was applied for the damped wave equation in \cite{FL19} or for the polyharmonic operator in \cite{FL20}.

When defining the spectral determinant using the spectral zeta function, complex logarithms of the eigenvalues appear (for more details see Section~\ref{sec:def}). However, the complex logarithm is not a unique function. When one wants to define it as a unique function, one must choose a certain branch -- an interval of the width $2\pi$ from which the arguments of the eigenvalues are taken. As it was already mentioned in \cite{QHS93, FL19}, the choice of the branch may influence the value of the determinant. The aim of the current note is to shed some light on this problem. For various distributions of the eigenvalues in the complex plane and different choices of the branch cut, we find how the determinant changes when altering the branch cut.

The result in \cite{FL19} obtained for the linear distribution of the eigenvalues on the imaginary axis is generalized in two ways: we allow for multiple rays along which the eigenvalues are distributed and we generalize the result to power growth. For this setting, we prove that the determinant changes the sign when the branch cut crosses one of the rays on which the eigenvalues are distributed. Moreover, we prove that for the exponential and the logarithmic growth the spectral determinant cannot be reasonably defined. Finally, we study the distribution of the eigenvalues on a line not going through the origin and we compare the results to the previous results on the damped wave equation.

The paper is structured as follows. In the second section, we properly define all notions used in the paper. In Section~\ref{sec:dwe}, we give an introductory example showing how the choice of the branch cut influences the determinant. Section~\ref{sec:results} gives the main results of the paper; we compare the determinants for different branch cuts and various distributions of the eigenvalues.

\section{Definition and preliminaries}\label{sec:def}
Throughout this paper, we will assume an operator $A$ with discrete spectrum. Operator $A$ can be in general non-self-adjoint and thus its eigenvalues may not be real. To define the spectral determinant for this operator, we have to introduce the spectral zeta function first, which is a function of the complex parameter $s$.

\begin{definition}\label{def:zeta}
The spectral zeta function of the operator $A$ is
\begin{equation}
  \zeta_A(s) = \sum_{j=1}^\infty \lambda_j^{-s}\,,\label{eq:zeta}
\end{equation}
where $\lambda_j$'s are the eigenvalues of the operator $A$.
\end{definition}

This definition is a generalization of the Riemann zeta function $\zeta_\mathrm{R}(s) = \sum_{j=1}^\infty j^{-s}$. We stress that, similarly to the Riemann zeta function, the sum in the spectral zeta function is typically not convergent for all complex $s$. However, for the most common operators (as, for instance, the Sturm-Liouvile operators) the sum is typically convergent in the half-plane $\mathrm{Re}\,s>c$. 

Since it will be used in Subsection~\ref{sec:outside}, we also introduce the Hurwitz zeta function.
\begin{definition}
Hurwitz zeta function is a function of two complex parameters $s$ and $a$ defined by the formula
$$
  \zeta_\mathrm{H}(s,a) = \sum_{j=0}^\infty \frac{1}{(j+a)^s}\,.
$$  
\end{definition}

Now we can define the spectral determinant.

\begin{definition}\label{def:det}
The spectral determinant for the operator $A$ is defined as
\begin{equation}
  \mathrm{det}\,A = \mathrm{exp}(-\zeta_A'(0))\,, \label{eq:detdef}
\end{equation}
where prime denotes the complex derivative of the zeta function with respect to $s$.
\end{definition}

Notice that $s=0$ is the point where the sum \eqref{eq:zeta} is not convergent, as it consists of infinitely many ones. However, one can bypass this issue if the sum is properly defined in the above-mentioned half-plane. We use the fact that the function can be uniquely meromorphically continued from the half-plane to the rest of the complex plane and the needed derivative at zero is computed using this continuation. 

For the reader's convenience, we introduce the well-known notions from complex analysis (see, e.g. \cite{Rud86,Ahl79}).

\begin{definition}
We say that a complex function of the complex variable $f(z)$ is \emph{holomorphic} in an open set $\Omega\subset \mathbb{C}$ if there exists its complex derivative $f'(z)$ for every $z_0 \in \Omega$. The function which is holomorphic in $\Omega$ up to the set of isolated points is called \emph{meromorphic}. The function is \emph{complex analytic} at $z_0$ if it is infinitely many times differentiable and it is equal to its Taylor series in the neighbourhood of $z_0$.
\end{definition}

\begin{theorem}
For functions on an open ball, the function is holomorphic if and only if it is analytic.
\end{theorem}

\begin{theorem}
Two functions that are complex analytic in $\tilde \Omega$ and coincide on some set with an accumulation point in $\tilde \Omega$ are identical. 
\end{theorem}

Hence if we have a meromorphic function in an open set $\Omega$, we can meromorphically continue it to the whole complex plane. The zeta function in the Definition~\ref{def:det} is thus understood as the unique meromorphic continuation to zero from the above-mentioned set where the sum in the definition of the zeta function converges.

At the end of this section, let us mention a property that will be important in the following sections. There appears the term $\lambda_j^{-s}$ in the definition of the spectral zeta function. Let us stress that both $\lambda_j$ and $s$ are complex and that the expression can be rewritten as $\mathrm{exp\,}(-s \ln{(\lambda_j)})$. Hence the spectral zeta function (and through it also the spectral determinant) is dependent on the definition of the complex logarithm of the eigenvalues. The complex logarithm can be regarded as a multivalued function with the imaginary part having values that can differ by multiples of $2\pi$. When one wants to define the logarithm as a single-valued function, one must specify the interval of the width $2\pi$ and take the arguments of the numbers in the argument of the logarithm from that interval. Thus one chooses one particular \emph{branch} of the logarithm. Then the ray in the complex plane where the arguments are discontinuous is called the \emph{branch cut}.

As it was found in \cite{FL19} and earlier in \cite[Ex.~11]{QHS93}, the choice of the branch cut of the logarithm can influence the value of the spectral determinant. In the next sections, we will investigate this issue deeper and try to elucidate under which conditions the spectral determinant changes and how.

\section{Example -- damped wave equation}\label{sec:dwe}
In this section, we will reproduce an example from \cite{FL19}. In the mentioned paper, the spectral determinant for the damped wave equation on an interval of the length $T$ was studied. On this interval, the equation
\begin{equation}\label{eq:dwe1}
  \frac{\partial^2 v(t,x)}{\partial t^2} + 2 a(x) \frac{\partial v(t,x)}{\partial t} = \frac{\partial^2 v(t,x)}{\partial x^2}\,,
\end{equation}
is investigated, where $a(x)\in \mathcal{C}([0,T])$ is the damping function. The Dirichlet boundary conditions $v(0) = v(T) = 0$ and certain initial conditions are assumed. The problem can be rewritten into the form of a matrix equation
\begin{equation}\label{eq:dwe2}
    \frac{\partial}{\partial t} \begin{pmatrix}v_0(t,x)\\v_1(t,x)\end{pmatrix} = \begin{pmatrix}0 & 1 \\ \frac{\partial^2}{\partial x^2}& -2 a(x)\end{pmatrix}\begin{pmatrix}v_0(t,x)\\v_1(t,x)\end{pmatrix}\,.
\end{equation}
It is easy to find that after the substitution for $v_1$, the variable $v_0$ satisfies \eqref{eq:dwe1}. Furthermore, the ansatz $v_0 (t,x) = \mathrm{e}^{\lambda t} u_0(x)$, $v_1 (t,x) = \mathrm{e}^{\lambda t} u_1(x)$ leads to the alternative formulation of the problem -- finding the eigenvalues for the matrix operator
$A_{\mathrm{DWE}} = \begin{pmatrix}0 & 1 \\ \frac{\partial^2}{\partial x^2}& -2 a(x)\end{pmatrix}$. The eq.~\eqref{eq:dwe2} translates to
$$
  A_{\mathrm{DWE}}  \begin{pmatrix}u_0(x)\\ u_1(x)\end{pmatrix} = \lambda \begin{pmatrix}u_0(x)\\ u_1(x)\end{pmatrix}\,.$$

The spectral determinant for the operator $A_{\mathrm{DWE}}$ was found in \cite{FL19} and it was proven that it does not depend on the damping. 
The effect of the choice of the branch cut that is found in \cite{FL19} is visible already for the case without damping ($a(x) \equiv 0$), i.e. the operator with the eigenvalues $\lambda_{j\pm} = \pm \frac{j\pi}{T}i$, $j\in \mathbb{N}$. We will study the spectral determinant for the branch cuts in the negative and the positive real axis. For the former, the interval of the arguments of the eigenvalues will be chosen from the interval $(-\pi, \pi)$, and the eigenvalues can be rewritten as 
$$
  \lambda_{j+} = \frac{j\pi}{T} \mathrm{e}^{i\frac{\pi}{2}}\,,\quad \lambda_{j-} = \frac{j\pi}{T} \mathrm{e}^{-i\frac{\pi}{2}}\,.
$$
The spectral zeta function can be then written as
\begin{eqnarray}
  \zeta_A(s) &=& \sum_{j=1}^\infty \left[\left(\frac{j\pi}{T} \mathrm{e}^{i\frac{\pi}{2}}\right)^{-s}+\left(\frac{j\pi}{T} \mathrm{e}^{-i\frac{\pi}{2}}\right)^{-s}\right]\nonumber\\
             &=& \sum_{j=1}^\infty \left(\frac{j\pi}{T}\right)^{-s}(\mathrm{e}^{-i\frac{\pi}{2}s}+\mathrm{e}^{i\frac{\pi}{2}s}) \nonumber\\
             &=& \sum_{j=1}^\infty 2\left(\frac{j\pi}{T}\right)^{-s}\cos{\left(\frac{\pi s}{2}\right)}\nonumber\\
             &=& 2\mathrm{e}^{s\log{\left(\frac{T}{\pi}\right)}}\cos{\left(\frac{\pi s}{2}\right)} \zeta_\mathrm{R}(s)\,, \label{eq:dweleft}            
\end{eqnarray}
where $\zeta_\mathrm{R}$ is the Riemann zeta function.

The Riemann zeta function is well defined in the half-plane $\mathrm{Re\,}s > 1$ and thus can be meromorphically continued to the rest of the complex plane with the known values 
\begin{eqnarray}
 \zeta_{\mathrm{R}}(0) &=& -\frac{1}{2}\,,\label{eq:zeta0} \\
 \zeta'_{\mathrm{R}}(0) &=& -\frac{1}{2}\log{(2\pi)}\,.\label{eq:zetaprime0}
\end{eqnarray}
Differentiating \eqref{eq:dweleft} and using \eqref{eq:zeta0} and \eqref{eq:zetaprime0} we get
$$
  \zeta'_A(0) = 2\log{\frac{T}{\pi}}\zeta_{\mathrm{R}}(0)+2\zeta'_\mathrm{R}(0) = -\log{(2T)}\,.
$$
Hence the determinant is 
$$
  \mathrm{det}\,A = \mathrm{e}^{\log{(2T)}} = 2T\,.
$$

For the choice of the cut on the positive real axis, one has to take the arguments of the eigenvalues from the interval $(0,2\pi)$. Hence we plug into the formula for the zeta function the values 
$$
  \lambda_{j+} = \frac{j\pi}{T} \mathrm{e}^{i\frac{\pi}{2}}\,,\quad \lambda_{j-} = \frac{j\pi}{T} \mathrm{e}^{i\frac{3\pi}{2}}
$$
obtaining by similar manipulations as above the result
$$
    \zeta_A(s) = 2 \mathrm{e}^{-i \pi s} \mathrm{e}^{s\log{\left(\frac{T}{\pi}\right)}}\cos{\left(\frac{\pi s}{2}\right)} \zeta_\mathrm{R}(s)\,.           
$$
Thus the derivative at zero differs by the factor of $i\pi $.
$$
    \zeta'_A(0) =-2 i \pi \zeta_\mathrm{R}(0) + 2\log{\frac{T}{\pi}}\zeta_{\mathrm{R}}(0)+2\zeta'_\mathrm{R}(0) = i \pi -\log{(2T)}\,.           
$$
This results in the spectral determinant that has a different sign.
$$
  \mathrm{det}\,A = \mathrm{e}^{-i\pi+\log{(2T)}} = -2T\,.
$$

This result shows that the change of the branch cut can in some cases cause the difference in the spectral determinant, for instance, in its sign. The analysis provided in this section illustrates that to get a proper value of the spectral determinant, one must clearly specify which branch of the logarithm is assumed in its definition. Even in this simple example inspired by a physical problem, the value of the spectral determinant differs; this drives us to investigate this problem further and find how the determinant changes in different geometrical settings of the eigenvalues of the problem.

\section{Results}\label{sec:results}
In this section, we try to generalize this result for different and more general types of operators, or in other words, distributions of eigenvalues. We start with a result already mentioned in \cite{FL19} concerning the shift of the branch cut through finitely many eigenvalues.

\begin{theorem}\label{thm:finitely}
If the branch cut of the logarithm moves so that it crosses finitely many eigenvalues of the operator $A$, its spectral determinant does not change.
\end{theorem}

We generalize the result of Sec.~\ref{sec:dwe} in two directions. First, the eigenvalue distances from the origin grow at a different rate (as a power of $j$, exponentially or logarithmically). Secondly, we allow for more rays on which the eigenvalues are situated.

\subsection{Power growth}\label{sec:power}
One of the main results of the paper is the following theorem concerning the power growth of the eigenvalues.

\begin{theorem}\label{thm:power}
Let us have finite number of angles $\alpha_\ell$, such that $0<\alpha_1<\alpha_2<\dots<\alpha_N<2\pi$. Let us assume that the eigenvalues of the operator $A$ are
\begin{eqnarray*}
\lambda_j^{(1)} = c_1 j^{c_2} \mathrm{e}^{i\alpha_1}\,,\\
\lambda_j^{(2)} = c_1 j^{c_2} \mathrm{e}^{i\alpha_2}\,,\\
\vdots\\
\lambda_j^{(N)} = c_1 j^{c_2} \mathrm{e}^{i\alpha_N}\,,
\end{eqnarray*}
(see Figure~\ref{fig:power}), where $c_1, c_2 \in \mathbb{R}_+$. If the branch cut is moved so that it passes $n$ half-lines with angles $\alpha_i$, the determinant will be $(-1)^n$-multiple of the former determinant.
\end{theorem}
\begin{figure}
\centering
  \includegraphics[width=8cm]{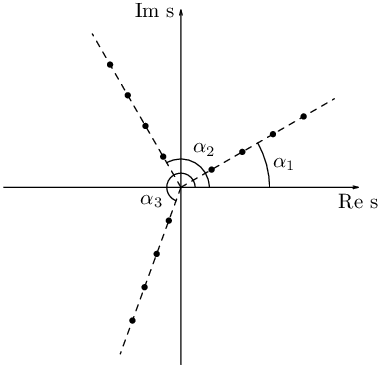}
  \caption{Power growth of the eigenvalues on more half-lines.}
  \label{fig:power}
\end{figure}%
\begin{proof}
Let us define angles $\beta_\ell$, $\ell=1,\dots, N+1$ such that $0<\beta_1<\alpha_1<\beta_2<\alpha_2<\beta_3<\dots<\beta_N<\alpha_N<\beta_{N+1}<2\pi$ and assume the branch cuts on the half-lines under the angles $\beta_\ell$. For a particular branch cut under the angle $\beta_\ell$, the arguments of the eigenvalues will be in the interval $(\beta_\ell,\beta_\ell+2\pi)$. It is clear that the arguments of the eigenvalues in the particular rays are successively $\alpha_1+2\pi$, $\alpha_2+2\pi$, \dots , $\alpha_{\ell-1}+2\pi$, $\alpha_{\ell}$,  $\alpha_{\ell+1}$, \dots ,  $\alpha_{N}$. The spectral zeta function is 
\begin{eqnarray*}
  \zeta_A (s) &=& \sum_{j=1}^\infty (c_1 j^{c_2})^{-s} [\mathrm{e}^{-i \alpha_1 s- 2\pi i s}+\mathrm{e}^{-i \alpha_2 s- 2\pi i s}+\dots + \mathrm{e}^{-i \alpha_{\ell-1} s- 2\pi i s}+ \mathrm{e}^{-i \alpha_{\ell} s}+\dots + \mathrm{e}^{-i \alpha_{N} s}]\\
  &=& \mathrm{e}^{-s \log{c_1}} \zeta_{\mathrm{R}}(c_2 s) [\mathrm{e}^{-i \alpha_1 s- 2\pi i s}+\mathrm{e}^{-i \alpha_2 s- 2\pi i s}+\dots + \mathrm{e}^{-i \alpha_{\ell-1} s- 2\pi i s}+ \mathrm{e}^{-i \alpha_{\ell} s}+\dots + \mathrm{e}^{-i \alpha_{N} s}]\,.
\end{eqnarray*}
The sum is convergent for $\mathrm{Re}\,s>c_2^{-1}$, hence we can continue it to the rest of the complex plane. The derivative of the spectral zeta function at zero is 
\begin{eqnarray*}
 \zeta'_A (0) &=& -\log{(c_1)}\, \zeta_{\mathrm{R}}(0) N + c_2 \zeta'_{\mathrm{R}}(0) N + \zeta_{\mathrm{R}}(0) (-i)\left[2\pi (\ell-1)+\sum_{k=1}^N \alpha_k \right]\\
 &=& \frac{1}{2}N\log{(c_1)} -\frac{1}{2}N c_2\log{(2\pi)} + i\pi (\ell-1)+ i \frac{1}{2} \sum_{k=1}^N \alpha_k\,.
\end{eqnarray*}
The determinant is then according to \eqref{eq:detdef} equal to 
\begin{equation*}
 \mathrm{det}\,A = (-1)^{\ell-1} c_1^{-\frac{N}{2}} (2\pi)^{\frac{c_2 N}{2}} \mathrm{exp}\left[- \frac{i}{2} \sum_{k=1}^N \alpha_k\right]\,.
\end{equation*}
If the branch cut is changed so that it crosses $n$ rays of eigenvalues, i.e. we move from the index $\ell$ to the index $\ell+n$, the ratio of the determinants is clearly from the previous formula $(-1)^n$ and hence we have
$$
  \mathrm{det}\,A_{\ell+n} = \mathrm{det}\,A_{\ell} (-1)^n \,,
$$
where $\mathrm{det}\,A_{\ell}$ is the determinant for the branch cut under the angle $\beta_\ell$ and $\mathrm{det}\,A_{\ell+n}$ is the determinant for the branch cut under the angle $\beta_{\ell+n}$. 
\end{proof}

Let us stress that even for the cut on the same place in the complex plane but with the angle different by $2\pi$ the determinant does not have to be the same. If the number of rays of the eigenvalues is odd, the determinant changes the sign.

\subsection{Exponential growth}
This subsection is devoted to the exponential behaviour of the eigenvalues.

\begin{theorem}\label{thm:exponential}
If a ray of eigenvalues behaving as $\lambda_j = c_1 \mathrm{e}^{c_2 j}\mathrm{e}^{i\alpha}$, $c_1, c_2\in \mathbb{R}_+$ is present, the spectral determinant diverges to $+\infty$. 
\end{theorem}
\begin{proof}
We will restrict ourselves to the case when there is only the above-mentioned ray of eigenvalues, although the presence of other eigenvalues (either finitely many or infinitely many with power growth) does not influence the result. First, we write down the spectral zeta function.
\begin{eqnarray*}
  \zeta_A (s) &=& \sum_{j=1}^\infty c_1^{-s} \mathrm{e}^{-c_2 j s}\mathrm{e}^{-i\alpha s}\\
  &=& c_1^{-s} \mathrm{e}^{-i\alpha s}\sum_{j=1}^\infty \mathrm{e}^{-c_2 j s}\\
  &=& c_1^{-s} \mathrm{e}^{-i\alpha s} \frac{\mathrm{e}^{-c_2 s}}{1-\mathrm{e}^{-c_2 s}}\\
  &=& \mathrm{e}^{-s\log{(c_1)}} \mathrm{e}^{-i\alpha s} \frac{1}{\mathrm{e}^{c_2 s}-1}\,,\\  
\end{eqnarray*}
where we have used the sum for the geometric series that converges for $\mathrm{Re\,}s>0$. Its derivative is
\begin{eqnarray*}
  \zeta'_A (s) &=& -\log{c_1} \mathrm{e}^{-s\log{(c_1)}} \mathrm{e}^{-i\alpha s} \frac{1}{\mathrm{e}^{c_2 s}-1}-i\alpha \mathrm{e}^{-s\log{(c_1)}} \mathrm{e}^{-i\alpha s} \frac{1}{\mathrm{e}^{c_2 s}-1}\\
&&-c_2 \mathrm{e}^{-s\log{(c_1)}} \mathrm{e}^{-i\alpha s} \frac{\mathrm{e}^{c_2 s}}{(\mathrm{e}^{c_2 s}-1)^2}\,.\\  
\end{eqnarray*}
Both $\zeta_A$ and $\zeta'_A$ diverge as $s=0$; the limit of the derivative at zero is $-\infty$. Hence the spectral determinant diverges to $+\infty$.
\end{proof}

\subsection{Logarithmic growth}
In this subsection, we study the case when the eigenvalues on a ray grow logarithmically or slower. 

\begin{theorem}\label{thm:logarithmic}
Let $\log$ denote the natural logarithm. If a ray of eigenvalues behaving as $\lambda_j = \omega_j\mathrm{e}^{i\alpha}$, $\omega_j\in \mathbb{R}_+$, $\omega_j\leq c_1 \log{(c_2 j)}$, $c_1\in \mathbb{R}_+$, $c_2\geq 1$ is present, the spectral zeta function is not defined.
\end{theorem}
\begin{proof}
The spectral zeta function can be constructed similarly to the previous cases.
\begin{equation*}
  \zeta_A (s) = \mathrm{e}^{-i\alpha s} \sum_{j=1}^\infty \omega_j^{-s}\,.
\end{equation*}
Now we prove that $\sum_{j=1}^\infty \omega_j^{-s}$ diverges for all $s\in \mathbb{R}$ with $s>0$ (for $s\leq 0$ the claim is obvious).
\begin{equation*}
  \sum_{j=1}^\infty \omega_j^{-s} \geq c_1^{-s} \sum_{j=1}^\infty  (\log{(c_2 j)})^{-s} = c_1^{-s} \sum_{j=1}^\infty (\log{(c_2)}+\log{(j)})^{-s}\,.
\end{equation*}
The last sum diverges by the integral criterion. For any given $s>0$ there exists $c>1$ such that $\frac{1}{x}<\log{(x)}^{-s}$ for $x>c$. Moreover, there exists $\tilde c >0$ such that $(\log{(c_2)}+\log{(x)})^{-s}\geq (2\log{x})^{-s}\geq \tilde c (\log{x})^{-s}$ for $x>c$. Then
\begin{eqnarray*}
  \int_1^\infty (\log{(c_2)}+\log{(x)})^{-s} \,\mathrm{d}x &\geq &\int_c^\infty (\log{(c_2)}+\log{(x)})^{-s}\,\mathrm{d}x 
\\
&\geq &   \int_c^\infty \tilde c \log{(x)}^{-s} \,\mathrm{d}x > \int_c^\infty \frac{\tilde c}{x} \,\mathrm{d}x= \infty 
\end{eqnarray*}
and hence the sum diverges. Therefore, there is no half-plane $\mathrm{Re\,}s>\mathrm{const}.$ such that the zeta function is defined in this half-plane and hence the spectral determinant cannot be reasonably defined.
\end{proof}

\subsection{Eigenvalues on a vertical line outside the origin}\label{sec:outside}
The next example illustrates that the ratio of the determinants for different branch cuts is not always 1 or $-1$. We consider the example from Section~\ref{sec:dwe} with the eigenvalues shifted horizontally. The eigenvalues will be $\lambda_j = b+ij$, $j\in \mathbb{Z}\backslash\{0\}$ (see Figure~\ref{fig:outside}).  

\begin{figure}
\centering
  \includegraphics[width=8cm]{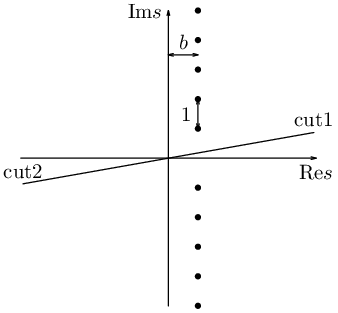}
  \caption{Eigenvalues on a vertical line outside the origin.}
  \label{fig:outside}
\end{figure}%

We will choose the branch cut in the following way. The first one will be on the positive real axis or just above it, hence the arguments of the eigenvalues will be taken from the interval $(0,2\pi)$. The second branch cut will be chosen on the negative real axis and the arguments of the complex numbers will be taken from the interval $(\pi, 3\pi)$. Note that the eigenvalues in the lower half-plane will have the same arguments for both choices of the branch cut. 

The spectral zeta functions can be written using the Hurwitz zeta functions in the following way.
\begin{equation}
  \zeta_1(s) = \mathrm{e}^{-\frac{\pi}{2}is} \zeta_\mathrm{H}(s,1-ib)+ \mathrm{e}^{-\frac{3\pi}{2}is}\zeta_\mathrm{H}(s,1+ib)\label{eq:ouside_zeta1}
\end{equation}
for the first choice of the branch cut and
\begin{equation}  
\zeta_2(s) = \mathrm{e}^{-\frac{5\pi}{2}is} \zeta_\mathrm{H}(s,1-ib)+ \mathrm{e}^{-\frac{3\pi}{2}is}\zeta_\mathrm{H}(s,1+ib)\label{eq:ouside_zeta2}
\end{equation}
for the second choice. For this construction, we have rotated the set of points $j\pm ib$, $j=1, \dots, \infty$ by corresponding angles. Note that in the definition of the Hurwitz zeta function, the sum goes from 0, while in the definition of the spectral zeta function it starts from 1. This results in the factor of 1 in the second argument of $\zeta_\mathrm{H}$.

Differentiating the expressions \eqref{eq:ouside_zeta1} and \eqref{eq:ouside_zeta2} one obtains
\begin{eqnarray*}
  \zeta_1'(0) = -\frac{\pi}{2}i \zeta_\mathrm{H}(0,1-ib) + \left.\frac{\partial \zeta_\mathrm{H}(s,1-ib)}{\partial s}\right|_{s=0}  - \frac{3\pi}{2}i \zeta_\mathrm{H}(0,1+ib)+ \left.\frac{\partial \zeta_\mathrm{H}(s,1+ib)}{\partial s}\right|_{s=0} \,,\\
  \zeta_2'(0) = -\frac{5\pi}{2}i \zeta_\mathrm{H}(0,1-ib) + \left.\frac{\partial \zeta_\mathrm{H}(s,1-ib)}{\partial s}\right|_{s=0}  - \frac{3\pi}{2}i \zeta_\mathrm{H}(0,1+ib)+ \left.\frac{\partial \zeta_\mathrm{H}(s,1+ib)}{\partial s}\right|_{s=0} \,,\\  
\end{eqnarray*}

Let us denote by $\mathrm{det}_1$ and $\mathrm{det}_2$ the determinant for the first and the second cut, respectively. The ratio of both determinants is 
\begin{equation}
  \frac{\mathrm{det\,}_1}{\mathrm{det}_2} = \mathrm{e}^{-\zeta_1'(0)+\zeta_2'(0)} = \mathrm{e}^{-2i \pi \zeta_\mathrm{H}(0,1-ib)}\,.\label{eq:outside:ratio1}
\end{equation}

Our final task is to find the value of the Hurwitz zeta function $\zeta(0,1-ib)$. We can use the following formula (see, e.g. \cite[eq. (1.10.7)]{EMOT53})
$$
  \zeta_\mathrm{H}(s,a) = \frac{1}{2} a^{-s}+\frac{a^{1-s}}{s-1}+2\int_0^\infty \frac{\sin{(s \arctan{(x/a)})}}{(a^2+x^2)^{s/2}(\mathrm{e}^{2\pi x}-1)}\,\mathrm{d}x\,.
$$
Applying this formula for $s=0$ and $a=1-ib$ one finds that the sine and hence also the integral vanish and we obtain
$$
  \zeta_\mathrm{H}(0,1-ib) = \frac{1}{2}+\frac{1-ib}{-1} = -\frac{1}{2}+ bi\,.
$$  
Substitution to \eqref{eq:outside:ratio1} yields
$$
  \frac{\mathrm{det\,}_1}{\mathrm{det}_2} = \mathrm{e}^{-2i \pi(-\frac{1}{2}+ bi)} = -\mathrm{e}^{2b\pi}\,.
$$

\begin{remark}
We can see that unlike in \cite{FL19} the determinant not only changes the sign, but the ratio of the determinants depends on the parameter $b$. Another example can be found in \cite[Ex. 11]{QHS93}. A more detailed analysis of the eigenvalues of our system and the damped wave equation in \cite{FL19} for the general value of the damping shows that, on one hand, the eigenvalue distribution looks very similar and, on the other hand, there are differences in the higher terms of the eigenvalue asymptotics. For the damped wave equation on the interval of length 1 with the damping $a(x)$, the large $j$ eigenvalue asymptotics is 
$$
  \lambda_j \approx \pi j i- \left<a\right> +\frac{\left<a^2\right>}{2\pi i j}+ \frac{1}{2\pi^2 j^2}[\left<a^3\right>-\left<a\right>\left<a^2\right>+\frac{1}{2}(a'(1)-a'(0))]+\dots 
$$
(see \cite{BF09}). The correct choice of the length of the interval supporting the damped wave equation (in particular, equal to $\pi$)  and $b=-\left<a\right>$ gives the same first two terms of the asymptotics as in our example. However, the higher-order terms differ. This results in different behaviour of the determinant; in \cite{FL19} the determinant did not depend on the damping, while in the present example, it depends exponentially. One can deduce that even small changes in the eigenvalue asymptotics can influence the spectral determinant.
\end{remark}

\section{Conclusions}
In this paper, we illustrated the subtleties of the spectral determinant and the spectral zeta function. The spectral zeta function, used for defining the spectral determinant, is a function of a complex variable $s$ and (in general infinitely many) complex eigenvalues $\lambda_j$. In the variable $s$, one may find a certain region in which the infinite sum in its definition converges, the zeta function is well-defined and may be meromorphically continued into the rest of the $s$-complex plane. On the other hand, the situation in the complex plane in which the eigenvalues $\lambda_j$ ``live'' is more complicated. The value of the spectral zeta function depends on the choice of the interval from which the arguments of the eigenvalues are taken. This may result in the discontinuities of the spectral zeta function in the $\lambda$-plane. If the branch cut moves through infinitely many eigenvalues (or, from the other perspective, if we perturb the eigenvalues so that they move through the branch cut), the spectral zeta function (and hence also the spectral determinant) may change.

Although this phenomenon was mentioned earlier in \cite{FL19, QHS93}, to the best of our knowledge, it has not been studied in detail previously. In the current paper, we gave new theorems~\ref{thm:power}, \ref{thm:exponential}, and \ref{thm:logarithmic} that find for eigenvalues on rays with power, exponential and logarithmic behaviour whether the zeta function and the spectral determinant can be defined and if yes, how the determinant changes when moving the branch cut. Moreover, in the Subsection~\ref{sec:outside} we study in detail the example of an operator with the eigenvalues on the line not including the origin. Such a distribution of eigenvalues is close to the distribution of eigenvalues of the damped wave equation on an interval. However, we show that the behaviour of the spectral determinants for both operators differs significantly. Our result manifests that the higher terms of the asymptotics of the eigenvalues for the damped wave equation are crucial for the behaviour of some of its spectral properties, e.g. the spectral determinant.

\section*{Acknowledgements}
The authors were supported by the Czech Science Foundation within the project 22-18739S. We thank the reviewer for useful remarks that improved the presentation of the paper.

\end{document}